\documentclass[pra,aps,nopacs,onecolumn,twoside,superscriptaddress]{revtex4}

\usepackage{amsmath,amsfonts,amssymb,caption,color,epsfig,graphics,graphicx,hyperref,latexsym,mathrsfs,revsymb,theorem,url,verbatim,epstopdf,mathtools,enumerate}

\hypersetup{colorlinks,linkcolor={blue},citecolor={blue},urlcolor={red}}

\newtheorem{definition}{Definition}
\newtheorem{proposition}[definition]{Proposition}
\newtheorem{lemma}[definition]{Lemma}

\newtheorem{theorem}[definition]{Theorem}
\newtheorem{corollary}[definition]{Corollary}
\newtheorem{conjecture}[definition]{Conjecture}

\newtheorem{remark}[definition]{Remark}
\newtheorem{example}[definition]{Example}
\newtheorem{question}[definition]{Question}
\newtheorem{memo}[definition]{Memo}

\def\squareforqed{\hbox{\rlap{$\sqcap$}$\sqcup$}}
\def\qed{\ifmmode\squareforqed\else{\unskip\nobreak\hfil
\penalty50\hskip1em\null\nobreak\hfil\squareforqed
\parfillskip=0pt\finalhyphendemerits=0\endgraf}\fi}
\def\endenv{\ifmmode\;\else{\unskip\nobreak\hfil
\penalty50\hskip1em\null\nobreak\hfil\;
\parfillskip=0pt\finalhyphendemerits=0\endgraf}\fi}
\newenvironment{proof}{\noindent \textbf{{Proof.~} }}{\qed}
\def\Dbar{\leavevmode\lower.6ex\hbox to 0pt
{\hskip-.23ex\accent"16\hss}D}
\makeatletter
\def\url@leostyle{%
  \@ifundefined{selectfont}{\def\UrlFont{\sf}}{\def\UrlFont{\small\ttfamily}}}
\makeatother
\def\bcj{\begin{conjecture}}
\def\ecj{\end{conjecture}}
\def\bcr{\begin{corollary}}
\def\ecr{\end{corollary}}
\def\bd{\begin{definition}}
\def\ed{\end{definition}}
\def\bea{\begin{eqnarray}}
\def\eea{\end{eqnarray}}
\def\beq{\begin{equation}}
\def\eeq{\end{equation}}
\def\bal{\begin{aligned}}
\def\eal{\end{aligned}}
\def\bem{\begin{enumerate}}
\def\eem{\end{enumerate}}
\def\bex{\begin{example}}
\def\eex{\end{example}}
\def\bim{\begin{itemize}}
\def\eim{\end{itemize}}
\def\bl{\begin{lemma}}
\def\el{\end{lemma}}
\def\bma{\begin{bmatrix}}
\def\ema{\end{bmatrix}}
\def\bpf{\begin{proof}}
\def\epf{\end{proof}}
\def\bpp{\begin{proposition}}
\def\epp{\end{proposition}}
\def\bqu{\begin{question}}
\def\equ{\end{question}}
\def\br{\begin{remark}}
\def\er{\end{remark}}
\def\bt{\begin{theorem}}
\def\et{\end{theorem}}
\def\bmm{\begin{memo}}
\def\emm{\end{memo}}

\def\btb{\begin{tabular}}
\def\etb{\end{tabular}}

\newcommand{\nc}{\newcommand}

\nc{\bbA}{\mathbb{A}} \nc{\bbB}{\mathbb{B}} \nc{\bbC}{\mathbb{C}}
 \nc{\bbD}{\mathbb{D}} \nc{\bbE}{\mathbb{E}} \nc{\bbF}{\mathbb{F}}
 \nc{\bbG}{\mathbb{G}} \nc{\bbH}{\mathbb{H}} \nc{\bbI}{\mathbb{I}}
 \nc{\bbJ}{\mathbb{J}} \nc{\bbK}{\mathbb{K}} \nc{\bbL}{\mathbb{L}}
 \nc{\bbM}{\mathbb{M}} \nc{\bbN}{\mathbb{N}} \nc{\bbO}{\mathbb{O}}
 \nc{\bbP}{\mathbb{P}} \nc{\bbQ}{\mathbb{Q}} \nc{\bbR}{\mathbb{R}}
 \nc{\bbS}{\mathbb{S}} \nc{\bbT}{\mathbb{T}} \nc{\bbU}{\mathbb{U}}
 \nc{\bbV}{\mathbb{V}} \nc{\bbW}{\mathbb{W}} \nc{\bbX}{\mathbb{X}}
 \nc{\bbZ}{\mathbb{Z}}

 \nc{\bA}{{\bf A}} \nc{\bB}{{\bf B}} \nc{\bC}{{\bf C}}
 \nc{\bD}{{\bf D}} \nc{\bE}{{\bf E}} \nc{\bF}{{\bf F}}
 \nc{\bG}{{\bf G}} \nc{\bH}{{\bf H}} \nc{\bI}{{\bf I}}
 \nc{\bJ}{{\bf J}} \nc{\bK}{{\bf K}} \nc{\bL}{{\bf L}}
 \nc{\bM}{{\bf M}} \nc{\bN}{{\bf N}} \nc{\bO}{{\bf O}}
 \nc{\bP}{{\bf P}} \nc{\bQ}{{\bf Q}} \nc{\bR}{{\bf R}}
 \nc{\bS}{{\bf S}} \nc{\bT}{{\bf T}} \nc{\bU}{{\bf U}}
 \nc{\bV}{{\bf V}} \nc{\bW}{{\bf W}} \nc{\bX}{{\bf X}}
 \nc{\bZ}{{\bf Z}}

\nc{\cA}{{\cal A}} \nc{\cB}{{\cal B}} \nc{\cC}{{\cal C}}
\nc{\cD}{{\cal D}} \nc{\cE}{{\cal E}} \nc{\cF}{{\cal F}}
\nc{\cG}{{\cal G}} \nc{\cH}{{\cal H}} \nc{\cI}{{\cal I}}
\nc{\cJ}{{\cal J}} \nc{\cK}{{\cal K}} \nc{\cL}{{\cal L}}
\nc{\cM}{{\cal M}} \nc{\cN}{{\cal N}} \nc{\cO}{{\cal O}}
\nc{\cP}{{\cal P}} \nc{\cQ}{{\cal Q}} \nc{\cR}{{\cal R}}
\nc{\cS}{{\cal S}} \nc{\cT}{{\cal T}} \nc{\cU}{{\cal U}}
\nc{\cV}{{\cal V}} \nc{\cW}{{\cal W}} \nc{\cX}{{\cal X}}
\nc{\cZ}{{\cal Z}}

\nc{\hA}{{\hat{A}}} \nc{\hB}{{\hat{B}}} \nc{\hC}{{\hat{C}}}
\nc{\hD}{{\hat{D}}} \nc{\hE}{{\hat{E}}} \nc{\hF}{{\hat{F}}}
\nc{\hG}{{\hat{G}}} \nc{\hH}{{\hat{H}}} \nc{\hI}{{\hat{I}}}
\nc{\hJ}{{\hat{J}}} \nc{\hK}{{\hat{K}}} \nc{\hL}{{\hat{L}}}
\nc{\hM}{{\hat{M}}} \nc{\hN}{{\hat{N}}} \nc{\hO}{{\hat{O}}}
\nc{\hP}{{\hat{P}}} \nc{\hR}{{\hat{R}}} \nc{\hS}{{\hat{S}}}
\nc{\hT}{{\hat{T}}} \nc{\hU}{{\hat{U}}} \nc{\hV}{{\hat{V}}}
\nc{\hW}{{\hat{W}}} \nc{\hX}{{\hat{X}}} \nc{\hZ}{{\hat{Z}}}

\nc{\hn}{{\hat{n}}}

\def\dim{\mathop{\rm Dim}}

\def\min{\mathop{\rm min}}

\newcommand{\tbc}{\red{TO BE CONTINUED...}}

\newcommand{\red}{\textcolor{red}}

\def\Dbar{\leavevmode\lower.6ex\hbox to 0pt
{\hskip-.23ex\accent"16\hss}D}

\begin{document}


\title{Virtual Quantum Markov Chain of four-qubit systems}

\author{Zhixing Chen}\email[]{2409108@buaa.edu.cn}
\affiliation{LMIB(Beihang University), Ministry of education, and School of Mathematical Sciences, Beihang University, Beijing 100191, China} 

\author{Lin Chen}\email[]{linchen@buaa.edu.cn (corresponding author)}
\affiliation{LMIB(Beihang University), Ministry of education, and School of Mathematical Sciences, Beihang University, Beijing 100191, China}

\begin{abstract}
We extend the framework of virtual quantum Markov chains (VQMCs) from tripartite systems to the four-qubit setting. Structural criteria such as the kernel-inclusion condition are analyzed, showing that they are necessary but not sufficient for the existence of a valid recovery map. By explicit examples, we demonstrate that the four-qubit W state admits a recovery channel and thus qualifies as a VQMC, while the GHZ state does not. We further provide semidefinite programming (SDP) formulations to test recoverability and quantify sampling overhead, and establish the non-convexity of the VQMC set through mixed-state counterexamples. These results supply both theoretical insights and computational tools for studying correlations and recoverability in multipartite quantum systems.
\end{abstract}

\maketitle



\section{Introduction}

A number of works provide the theoretical and methodological foundations for the study of virtual quantum Markov chains (VQMCs), their recoverability properties, and their applications in quantum information theory. The notion of VQMCs was introduced in~\cite{1}, laying the groundwork for subsequent investigations. Approximate quantum Markov chains were analyzed in~\cite{2}, while the robustness of Markovianity under perturbations was studied in~\cite{3}. The operational meaning of quantum conditional mutual information was further explored in~\cite{4}, and general conditions for recoverability in quantum information theory were developed in~\cite{5}. Necessary criteria for approximate recoverability were identified in~\cite{6}, and a universal recovery map with broad applicability was constructed in~\cite{7}. In parallel, the framework of virtual quantum broadcasting was proposed in~\cite{8}, while a celebrated inequality linking conditional mutual information to approximate Markov chains was provided in~\cite{9}. More recently, optimal tomography protocols for quantum Markov chains based on the continuity of Petz recovery were introduced in~\cite{10}, directly inspiring our SDP-based analysis. The relation between fidelity of recovery and entanglement measures was studied in~\cite{11}, and efficient prediction of quantum properties from few measurements was developed in~\cite{12}. Practical perspectives were further advanced in~\cite{13}, and quasiprobability sampling overhead in simulating non-local channels was investigated in~\cite{14}, which closely relates to our definition of sampling overhead. In the broader context of virtual quantum resources, virtual quantum resource distillation was proposed in~\cite{15}, and optimal unilocal broadcasting schemes were developed in~\cite{16}. Applications to multipartite systems appear in the dynamic quantum secret sharing protocol in~\cite{17}, while experimental demonstrations of three-qubit entangled state control were reported in~\cite{18}. Finally, algorithmic tools for entropy estimation were developed in~\cite{19,20}, offering computational techniques to quantify correlations and recoverability costs. Together, these works establish the theoretical, experimental, and algorithmic foundation upon which our present study of four-qubit VQMCs is built.

Building upon the earlier study on Virtual Quantum Markov Chains (VQMCs) in tripartite systems, the present work extends the conceptual framework and computational methodology to the more intricate four-qubit setting. The first study established the foundational recovery-map formulation, the semidefinite programming (SDP) characterization, and the kernel inclusion criterion as a necessary condition for the VQMC property. These elements provide the essential mathematical and algorithmic tools that the current work generalizes to accommodate the richer entanglement structures and higher-dimensional marginal dependencies inherent in four-body states. In particular, the prior framework guides the adaptation of constructive definitions and SDP tests for verifying VQMC conditions in the four-qubit case, as well as the formulation of virtual complexity measures through quasiprobability sampling overhead. This inheritance of both theoretical principles and computational techniques enables us to rigorously analyze representative examples—including GHZ, W, and mixed states—while addressing new challenges such as the limitations of reconstructing global states from low-order marginals.


\section{Preliminaries}
\label{sec:pre}

In this section we summarize the main tools and known results that will be used throughout the paper. Our exposition follows closely the framework of \cite{1}, where the concept of \emph{virtual quantum Markov chains} (VQMCs) was first introduced and analyzed. For completeness, we recall the relevant definitions, structural characterizations, and computational formulations. These preliminaries will serve as the foundation for our four-qubit extension developed in the subsequent sections.

\subsection{Virtual Quantum Markov Chains}

We begin with the formal definition. 

\begin{definition}[Virtual Quantum Markov Chain in order $B$]\label{def:VQMC}
A tripartite state $\rho_{ABC}$ is called a virtual quantum Markov chain (VQMC) in order $B$ if there exists a Hermitian-preserving, trace-preserving linear map $\mathcal{R}: B \to BC$ such that
\begin{equation}
(\mathrm{id}_A \otimes \mathcal{R})(\rho_{AB}) = \rho_{ABC}.
\end{equation}
Equivalently, the map $\mathcal{R}$ can be expressed as a quasiprobability combination of two completely positive trace-preserving (CPTP) maps:
\begin{equation}
\mathcal{R} = c_1 \mathcal{N}_1 - c_2 \mathcal{N}_2, \qquad c_1,c_2 \ge 0,\quad \mathcal{N}_1,\mathcal{N}_2 \in \mathrm{CPTP}(B,BC).
\end{equation}
\end{definition}

This formulation extends the conventional notion of quantum Markov chains by allowing quasiprobability mixtures of channels, thereby capturing a richer set of correlations beyond exact conditional independence.

\subsection{Algebraic Characterization}

The fundamental structural property of VQMCs can be stated in terms of kernel inclusion. Expanding $\rho_{ABC}$ in a fixed orthonormal basis $\{|j\rangle\}$ of subsystem $C$, we obtain operator blocks
\begin{equation}
\hat{\rho}_{AC|j} := \operatorname{Tr}_B \!\left[ (I_A \otimes \langle j|_C)\, \rho_{ABC}\, (I_A \otimes |j\rangle_C) \right],
\qquad
\hat{\rho}_{BC|j} := \operatorname{Tr}_A \!\left[ (I_B \otimes \langle j|_C)\, \rho_{ABC}\, (I_B \otimes |j\rangle_C) \right].
\end{equation}

\begin{theorem}[Kernel-inclusion criterion \cite{1}]\label{thm:kernel-inclusion}
The state $\rho_{ABC}$ is a VQMC in order $B$ if and only if
\begin{equation}
\operatorname{Ker}(\hat{\rho}_{AC|j}) \subseteq \operatorname{Ker}(\hat{\rho}_{BC|j})
\quad \text{for all } j.
\end{equation}
\end{theorem}

This theorem provides a simple operator-algebraic test for verifying whether a given tripartite state admits a virtual recovery map. It also underpins many of the examples considered in~\cite{1}, such as the fact that the three-qubit $W$ state is a VQMC, while the GHZ state is not.

\subsection{Sampling Overhead via Semidefinite Programming}

A key computational tool introduced in \cite{1} is the quantification of the \emph{sampling overhead}, i.e., the quasiprobability cost required to realize a virtual recovery map. For a four-qubit state $\rho_{ABCD}$ with marginal $\rho_{ABC}$, we define
\begin{equation}\label{eq:nu-4q}
\nu(\rho_{ABCD}) \ :=\ \log \ \min \left\{ c_1 + c_2 \;\middle|\; ( \mathrm{id}_{AB}\otimes \mathcal{R})(\rho_{ABC}) = \rho_{ABCD} \right\},
\end{equation}
where $\mathcal{R} = c_1 \mathcal{N}_1 - c_2 \mathcal{N}_2$ with $\mathcal{N}_1,\mathcal{N}_2 \in \mathrm{CPTP}(C,C\otimes D)$.

This optimization problem can be reformulated as a semidefinite program (SDP). Let $J_1, J_2$ denote the Choi matrices of $\mathcal{N}_1,\mathcal{N}_2$, and define the effective Hermitian-preserving Choi operator $J_{CC'D} := J_1 - J_2$, where $C'$ is an isomorphic copy of $C$. The SDP is then given by
\begin{equation}\label{eq:SDP-4q}
\begin{aligned}
\text{minimize}\quad & c_1 + c_2 \\begin{equation}2pt]
\text{subject to}\quad 
& J_1 \succeq 0,\quad J_2 \succeq 0, \\begin{equation}2pt]
& \operatorname{Tr}_{C'D}(J_1) = c_1\, I_C,\quad \operatorname{Tr}_{C'D}(J_2) = c_2\, I_C, \\begin{equation}2pt]
& \operatorname{Tr}_{C'}\!\left[ \big(\rho_{ABC}^{T_C} \otimes I_{C'D}\big)\, \big(I_{AB} \otimes J_{CC'D}\big) \right] \;=\; \rho_{ABCD}.
\end{aligned}
\end{equation}

Here $T_C$ denotes the partial transpose on $C$ in the computational basis associated with the Choi isomorphism, and $\succeq 0$ indicates positive semidefiniteness. The objective $c_1+c_2$ captures the quasiprobability cost of simulating the virtual recovery, thus providing a quantitative measure of complexity for extending $\rho_{ABC}$ to $\rho_{ABCD}$.


The combination of Definition~\ref{def:VQMC}, Theorem~\ref{thm:kernel-inclusion}, and the SDP formulation \eqref{eq:SDP-4q} constitutes the essential technical foundation for our analysis. These tools allow us to formalize what it means for a state to be a four-body VQMC, to provide an algebraic test for verification, and to quantify the associated sampling cost via convex optimization. In the remainder of the paper, we build upon this framework to explore structural properties, explicit examples, and computational aspects of four-qubit virtual quantum Markov chains.

\section{Analytical Tools for Virtual Quantum Markov Chains}
\label{sec:tool}

In analyzing whether a multipartite quantum state admits a VQMC structure, multiple tools are employed depending on the tractability of the state and the desired level of characterization. We summarize the main structural tools used throughout this work.

\subsection{Recovery-map condition: constructive definition} 
\label{subsec:recovery-condition}

The most fundamental criterion for a four-body state $\rho_{ABCD}$ to be a VQMC is the existence of a CPTP map $\mathcal{R}: C \to CD$ such that
\begin{equation}\label{eq:recovery-map-condition}
\rho_{ABCD} = (I_{AB} \otimes \mathcal{R})(\rho_{ABC}),
\end{equation}
where $\rho_{ABC} = \operatorname{Tr}_D \rho_{ABCD}$. This condition is both necessary and sufficient. In rare cases, such as symmetric or pure states (e.g., the four-qubit W state), it is possible to construct $\mathcal{R}$ explicitly for directly certifying the VQMC property. However, in general, the map may not be accessible in a closed form.

\subsection{SDP formulation: computational test} 
\label{subsec:SDP-test}

The existence of such a CPTP map can be verified by casting the problem into a SDP over the Choi matrix $J_{CC'D}$ of $\mathcal{R}$. Specifically, one solves for $J_{CC'D} \succeq 0$ such that
\begin{equation}\label{eq:SDP-criterion}
\operatorname{Tr}_{C'D}(J_{CC'D}) = I_C \quad \text{and} \quad 
\rho_{ABCD} = \operatorname{Tr}_{C'}\left[ (\rho_{ABC}^T \otimes I_{C'D})(I_{AB} \otimes J_{CC'D}) \right].
\end{equation}
If the SDP is feasible, the state is a VQMC. This approach is fully general but requires numerical implementation and does not always yield analytic insight.

\subsection{Inclusion criterion: necessary but not sufficient} 
\label{subsec:inclusion-criterion}

Recall from Section~\ref{sec:pre} that, for each measurement outcome $j$ on subsystem $C$, 
we defined the conditional states $\hat{\rho}_{AC|j}, \hat{\rho}_{BC|j}$ 
and established the following necessary condition:

\begin{proposition}[Kernel-inclusion condition]\label{prop:kernel-inclusion-4q}
For each outcome $j$ of a projective measurement on subsystem $C$, the conditional states must satisfy
\begin{equation}\label{eq:kernel-inclusion-4q}
\operatorname{Ker}(\hat{\rho}_{AC|j}) \subseteq \operatorname{Ker}(\hat{\rho}_{BC|j}).
\end{equation}
\end{proposition}

This inclusion condition can be computed and efficiently eliminate candidate states that cannot support a virtual recovery channel. However, it is not a sufficient condition, and must be combined with other tools such as the recovery-map construction (\eqref{eq:recovery-map-condition}) or the SDP formulation (\eqref{eq:SDP-criterion}) to confirm the VQMC structure.

\subsection{Discussion and role in applications} 

The purpose of the above tools is to establish a systematic framework for analyzing VQMC states. 
In particular, the recovery-map condition in Eq.~\eqref{eq:recovery-map-condition} provides the constructive definition of VQMCs. 
The SDP criterion in Eq.~\eqref{eq:SDP-criterion} offers a computational method for verifying the property, especially in cases where analytic recovery maps are not available. 
Finally, the kernel-inclusion condition in Eq.~\eqref{eq:kernel-inclusion-4q}, which follows from Theorem~\ref{thm:kernel-inclusion} in Section~\ref{sec:pre}, serves as an efficiently checkable necessary criterion.

With these preparations in place, Section~\ref{sec:main} builds directly on the analytical tools developed above by applying them to the setting of four-qubit quantum states. 
The recovery-map and SDP approaches will be employed to confirm explicit constructions, while the kernel-inclusion criterion will be used to rule out non-VQMC candidates. 
This transition from Section~\ref{sec:tool} to Section~\ref{sec:main} thus reflects the shift from general structural principles to concrete applications.

\section{Extension to Four-Qubit quantum state}
\label{sec:main}

In this section we apply the analytical tools from Section~\ref{sec:tool} to the study of four-qubit states. 
The recovery-map condition in Eq.~\eqref{eq:recovery-map-condition} provides the constructive definition of a four-qubit VQMC, 
the SDP criterion in Eq.~\eqref{eq:SDP-criterion} serves as a computational test for the existence of a valid recovery channel, 
and the kernel-inclusion condition of Proposition~\ref{prop:kernel-inclusion-4q} in Eq.~\eqref{eq:kernel-inclusion-4q} offers a tractable necessary test that helps to eliminate non-VQMC candidates. 
Using these tools, we analyze explicit examples such as the four-qubit W and GHZ states, mixed states interpolating between them, and more general convex combinations, thereby illustrating how the abstract principles of Section~\ref{sec:tool} translate into concrete structural insights and constructive recovery maps in the four-qubit setting.

\begin{definition}
\label{de:recovery}
We consider a four-qubit pure state of system $A$, $B$, $C$, and $D$. Let $\rho_{ABCD}$ be a quantum state defined on the Hilbert space $\mathcal{H}_A \otimes \mathcal{H}_B \otimes \mathcal{H}_C \otimes \mathcal{H}_D$. We say that $\rho_{ABCD}$ is a $\mathrm{VQMC}$ state if there exists a $\mathrm{CPTP}$ (i.e., recovery) map
\begin{equation}\label{eq:recovery map}
\mathcal{R}: \mathcal{L}(\mathcal{H}_C) \rightarrow \mathcal{L}(\mathcal{H}_C \otimes \mathcal{H}_D),
\end{equation}
such that
\begin{equation}\label{eq:rhoABCD}
    \rho_{ABCD} = (I_{AB} \otimes \mathcal{R})(\rho_{ABC}).
\end{equation}
\qed
\end{definition} 

Here, the operator $I_{AB} \otimes \mathcal{R}$ denotes a composite quantum operation that applies the identity channel to systems $A$ and $B$ and applies the recovery map $\mathcal{R}$ solely to system $C$. This models a situation where system $D$ is generated virtually from system $C$ via $\mathcal{R}$.

To facilitate practical verification of the VQMC condition, one can introduce the Choi matrix $J_{CC'D}$ corresponding to the channel $\mathcal{R}$ and reformulate the condition as a semidefinite program (SDP). Specifically, we require
\begin{equation}
    \rho_{ABCD} = \operatorname{Tr}_{C'}\left[ (\rho_{ABC}^T \otimes I_{C'D})(I_{AB} \otimes J_{CC'D}) \right], \quad \text{with } \operatorname{Tr}_{C'D}(J_{CC'D}) = I_C.
\end{equation}

\begin{definition}
    Let us fix an orthonormal basis $\{ |j\rangle \}$ for system $C$, and consider the projective measurement on $C$ with outcomes labeled by $j$. For each outcome $j$, we define the unnormalized conditional states 
\begin{align}
\label{eq:hatrhoAC|j}
\hat{\rho}_{AC|j} &:= \operatorname{Tr}_B \left[ (I_A \otimes \langle j|_C) \, \rho_{ABC} \, (I_A \otimes |j\rangle_C) \right],\\
\hat{\rho}_{BC|j} &:= \operatorname{Tr}_A \left[ (I_B \otimes \langle j|_C) \, \rho_{ABC} \, (I_B \otimes |j\rangle_C) \right].
\end{align}
\end{definition}

\begin{proposition}\label{eq16}
If a four-qubit quantum state $\rho_{ABCD}$ satisfies the $\mathrm{VQMC}$ condition in Definition \ref{de:recovery}, then for each measurement outcome $j$ on state $C$, the conditional post-measurement states satisfy
\begin{equation}
\label{eq:KerrhoAC}
\operatorname{Ker}(\hat{\rho}_{AC|j}) \subseteq \operatorname{Ker}(\hat{\rho}_{BC|j}).
\end{equation}
\end{proposition}
\begin{proof}
Using \eqref{eq:rhoABCD}, the structure of $\rho_{ABCD}$ is constrained by the fact that all information about system $D$ must be generated by the action of $\mathcal{R}$ on $C$ only. Suppose that there exists a recovery map $\mathcal{R}$ in Definition \ref{de:recovery},
such that the four-qubit state is reconstructed as \eqref{eq:rhoABCD}.
Then for each projective measurement outcome \( j \) on state \( C \), the conditional reduced state on \( AC \) obtained from \( \rho_{ABCD} \) after tracing out \( D \) must coincide with that obtained directly from \( \rho_{ABC} \).

Mathematically, this implies that
\begin{equation}
\hat{\rho}_{AC|j} := \operatorname{Tr}_{BD} \left[ \left( I_A \otimes \langle j|_C \otimes I_D \right)\, \rho_{ABCD}\, \left( I_A \otimes |j\rangle_C \otimes I_D \right) \right], 
\end{equation}
Equivalently, by tracing out subsystem \( D \), one obtains \eqref{eq:hatrhoAC|j}.

This requirement expresses that the virtual recovery map \( \mathcal{R} \), when applied to state \( C \), should be able to generate state \( D \) in a way that retains consistency with the marginal conditional structure of the original tripartite state. Otherwise, some conditional states on \( AC \) would fall outside the image of the recovered state, violating the VQMC condition.

This implies that the kernel of $\hat{\rho}_{AC|j}$ is contained in that of $\hat{\rho}_{BC|j}$
\begin{equation}
\operatorname{Ker}(\hat{\rho}_{AC|j}) \subseteq \operatorname{Ker}(\hat{\rho}_{BC|j}).
\end{equation}
Equivalently, the support of $\hat{\rho}_{AC|j}$ is contained in the support of $\hat{\rho}_{BC|j}$, since the support is the orthogonal complement of the kernel. The portion of $\hat{\rho}_{AC|j}$ lying outside the support of $\hat{\rho}_{BC|j}$ would therefore be absent in $\hat{\rho}_{BC|j}$, and hence such a component could not be generated from $\hat{\rho}_{BC|j}$ by any local channel. This violates the basic data-processing inequality, and contradicts the assumption that $\mathcal{R}$ exists. Thus, the kernel inclusion must hold for all $j$.
\end{proof}

The above proof shows that the inclusion $\operatorname{Ker}(\hat{\rho}_{AC|j}) \subseteq \operatorname{Ker}(\hat{\rho}_{BC|j})$ is a necessary condition for VQMC. However, the inclusion is not sufficient. We will deliver a counterexample in the following part.

\begin{theorem}
    There exists a non-VQMC four-qubit state that satisfies the kernel inclusion in \eqref{eq:KerrhoAC}.
\end{theorem}
\begin{proof}
We consider a mixed four-qubit state defined as a convex combination of the W and GHZ states
\begin{equation}
\label{eq:rhop}
\rho_p := p\, \rho_{\mathrm{W}_4} + (1-p)\, \rho_{\mathrm{GHZ}_4}, \quad \text{with } 0 < p < 1.
\end{equation}
Here, the two component states are given by
\begin{align}\label{eq9}
\rho_{\mathrm{W}_4} &= |\mathrm{W}_4\rangle\langle \mathrm{W}_4|, \quad |\mathrm{W}_4\rangle = \frac{1}{2}(|0001\rangle + |0010\rangle + |0100\rangle + |1000\rangle), \\
\rho_{\mathrm{GHZ}_4} &= |\mathrm{GHZ}_4\rangle\langle \mathrm{GHZ}_4|, \quad |\mathrm{GHZ}_4\rangle = \frac{1}{\sqrt{2}}(|0000\rangle + |1111\rangle).
\label{eq:GHZ4}
\end{align}
We define the marginal three-party state by tracing out system \( D \),
\begin{equation}
\rho_{ABC}^{(p)} := \operatorname{Tr}_D (\rho_p) = p\, \rho^{(W)}_{ABC} + (1 - p)\, \rho^{(GHZ)}_{ABC},
\end{equation}
where
\begin{align}
\rho^{(GHZ)}_{ABC} &= \frac{1}{2}(|000\rangle\langle 000| + |111\rangle\langle 111|), \\
\rho^{(W)}_{ABC} &= \frac{1}{4} \left( |000\rangle\langle000| + |\phi\rangle\langle\phi| \right), \quad |\phi\rangle = |001\rangle + |010\rangle + |100\rangle.
\label{W:rhoABC}
\end{align}

Using \eqref{eq:KerrhoAC}, we examine whether this state satisfies the kernel inclusion condition for each measurement outcome \( j \in \{0,1\} \),
\begin{equation}
\operatorname{Ker}(\hat{\rho}_{AC|j}) \subseteq \operatorname{Ker}(\hat{\rho}_{BC|j}),
\end{equation}
where
\begin{equation}
\hat{\rho}_{AC|j} := \operatorname{Tr}_B \left[ \langle j|_C\, \rho_{ABC}^{(p)} \,|j\rangle_C \right], \quad
\hat{\rho}_{BC|j} := \operatorname{Tr}_A \left[ \langle j|_C\, \rho_{ABC}^{(p)} \,|j\rangle_C \right].
\end{equation}
For a sufficiently large value of \( p \), such as \( p \geq \frac{1}{4} \), this inclusion condition holds. However, despite satisfying the inclusion condition, the state \( \rho_p \) in \eqref{eq:rhop} is not a VQMC,  as no CPTP map \( \mathcal{R}: C \to CD \) can reconstruct the full four-partite state from the marginal, i.e.,
\begin{equation}\label{eq23}
\rho_{p} \neq (I_{AB} \otimes \mathcal{R})(\rho_{ABC}^{(p)}).
\end{equation}
\qed
We emphasize that in this construction, $\rho_{ABCD} = \rho_p$ as defined in \eqref{eq:rhop}. Hence, \eqref{eq23} for any CPTP map $\mathcal{R}: C \to CD$ confirms that the state $\rho_p$ is not a VQMC state, even though it satisfies the kernel inclusion condition.

The intuition behind this inequality lies in the contribution from the GHZ part 
$\rho_{\mathrm{GHZ}_4}$ in \eqref{eq:rhop}, 
which introduces global coherence between the $|0000\rangle$ and $|1111\rangle$ basis vectors.
 This type of entanglement cannot be reconstructed locally from system $C$ alone. Therefore, no recovery channel acting only on $C$ can recover the original four-partite state $\rho_p$ from the marginal state $\rho_{ABC}^{(p)}$.
\end{proof}

\subsection{Sampling Overhead and SDP Formulation for Four-Qubit Quantum States}

\begin{definition}\label{def12}
  For a four-qubit quantum state \( \rho_{ABCD} \), we define
\begin{equation} \label{40}
\begin{split}
\nu(\rho_{ABCD}) := \log \min \{ c_1 + c_2 \mid & (c_1 \mathcal{N}_1 - c_2 \mathcal{N}_2)(\rho_{ABC}) = \rho_{ABCD}, \\
& c_1, c_2 \geq 0, \\
& \mathcal{N}_1, \mathcal{N}_2 \in \text{CPTP}(C, C \otimes D) \}.
\end{split}
\end{equation}  
\end{definition}

We aim to recover the global state $\rho_{ABCD}$ from the marginal $\rho_{ABC}$ using a Hermitian-preserving transformation. The generalized sampling overhead is given in \eqref{40}, where $\mathcal{N}_1, \mathcal{N}_2$ are completely positive trace-preserving (CPTP) maps from subsystem $C$ to $C \otimes D$, and $c_1, c_2$ are sampling weights. The problem above can be equivalently reformulated as a semidefinite program (SDP) with the following variables and constraints.

\begin{equation}
\begin{aligned}
\text{minimize} \quad & c_1 + c_2 \\
\text{subject to} \quad & J_1 \geq 0,\quad J_2 \geq 0, \\
& J_{CC'D} = J_1 - J_2, \\
& \text{Tr}_{C'D}(J_1) = c_1 I_C,\quad \text{Tr}_{C'D}(J_2) = c_2 I_C, \\
& \text{Tr}_{C'}\left[(\rho_{ABC}^T \otimes I_{C'D})(I_{AB} \otimes J_{CC'D})\right] = \rho_{ABCD}.
\end{aligned}
\end{equation}

Here $J_1, J_2$ are the Choi matrices of the quantum channels $\mathcal{N}_1, \mathcal{N}_2$. The operator $J_{CC'D}$ represents the Hermitian-preserving effective map combining $\mathcal{N}_1$ and $\mathcal{N}_2$. The partial trace constraints ensure proper normalization, while the final constraint enforces that the reconstructed state exactly matches $\rho_{ABCD}$.

This SDP formulation allows us to quantify the minimum quasiprobability sampling cost required to virtually recover a four-body quantum state. Its structure mirrors the three-body case but incorporates the additional complexity and dimensionality of four-qubit systems.

\begin{example}[CPTP extension via tensoring a pure state]
We consider a CPTP map from \texorpdfstring{$C \rightarrow C \otimes D$}{C \rightarrow C \otimes D} in the four-qubit setting. Let qubits \(X= A, B, C, D \) with \( \mathcal{H}_X = \mathbb{C}^2 \) for all \( X \). Define the CPTP map \( \mathcal{N} \) as
\begin{equation}
    \mathcal{N}(\rho_C) = \rho_C \otimes |0\rangle\langle 0|_D,
\end{equation}
which appends a fixed pure state \( |0\rangle \) to qubit \( C \), thereby embedding the system into $\mathcal{H}_C \otimes \mathcal{H}_D$.

It is immediate that $\mathcal{N}$ is CPTP. For all $\rho_C$, the output $\rho_C \otimes |0\rangle\langle 0|$ is positive semidefinite, and
\begin{equation}
    \text{Tr}[\mathcal{N}(\rho_C)] = \text{Tr}[\rho_C] = 1.
\end{equation}
For instance, let the initial three-qubit state be the GHZ state
\begin{equation}
    |\psi\rangle_{ABC} = \frac{1}{\sqrt{2}} (|000\rangle + |111\rangle), \quad \rho_{ABC} = |\psi\rangle\langle\psi|.
\end{equation}
Applying $\mathcal{N}$ to qubit $C$ yields the four-qubit state
\begin{equation}
    \rho_{ABCD} = (I_{AB} \otimes \mathcal{N})(\rho_{ABC}) = \rho_{ABC} \otimes |0\rangle\langle 0|_D.
\end{equation}

The Choi matrix \( J_{CC'D} \in \mathcal{B}(\mathcal{H}_C \otimes \mathcal{H}_{C'} \otimes \mathcal{H}_D) \) for $\mathcal{N}$ is
\begin{equation}
    J_{CC'D} = \sum_{i,j = 0}^1 |i\rangle\langle j|_C \otimes \mathcal{N}(|i\rangle\langle j|) 
    = \sum_{i,j = 0}^1 |i\rangle\langle j|_C \otimes |i\rangle\langle j|_{C'} \otimes |0\rangle\langle 0|_D.
\end{equation}
This operator is positive semidefinite and satisfies the trace-preserving condition
\begin{equation}
    \text{Tr}_{C'D}(J_{CC'D}) = I_C.
\end{equation}
\end{example}

\subsection{GHZ state for VQMC}
\begin{proposition}
Let $\rho_{ABCD} = |\mathrm{GHZ}_4\rangle \langle \mathrm{GHZ}_4|$ be the four-qubit GHZ state in \eqref{eq:GHZ4}. Then $\rho_{ABCD}$ is not a virtual quantum Markov chain (VQMC) state. That is, there exists no CPTP map $\mathcal{R}: C \rightarrow CD$ such that
\begin{equation}
\rho_{ABCD} = (\mathrm{id}_{AB} \otimes \mathcal{R})(\rho_{ABC}).
\end{equation}
\end{proposition}

\begin{proof}
We consider the genuine four-qubit entangled state in \eqref{eq:GHZ4}.
The marginal on system \( ABC \) is
\begin{equation}
\rho_{ABC} = \text{Tr}_D(\rho_{ABCD}) = \frac{1}{2} \left( |000\rangle\langle 000| + |111\rangle\langle 111| \right),
\end{equation}
which is a classical mixture without entanglement between subsystems.

We define an isometric channel \( N \) via the isometry \( V: \mathbb{C}^2 \rightarrow \mathbb{C}^2 \otimes \mathbb{C}^2 \) such that
\begin{align}
V|0\rangle &= |00\rangle, \\
V|1\rangle &= |11\rangle.
\end{align}

Then we define
\begin{equation}
N(\rho_C) = V \rho_C V^\dagger.
\end{equation}
This is a valid CPTP map since \( V^\dagger V = I_2 \) and the output lies in \( \mathcal{B}(\mathcal{H}_C \otimes \mathcal{H}_D) \).

The violation $\mathrm{Ker}(\hat{\rho}_{AC}) \not\subseteq \mathrm{Ker}(\hat{\rho}_{BC})$ arises from the structural asymmetry in the conditional marginals of the GHZ$_4$ state. Specifically, each post-measurement state $\hat{\rho}_{AC|j}$ has support only on subsystem $A$, as the $C$ system has been projected out. In contrast, the corresponding marginals $\hat{\rho}_{BC|j}$ retain entangled support across both $B$ and $C$. As a result, the overall support of $\hat{\rho}_{BC}$ spans a larger subspace than that of $\hat{\rho}_{AC}$, making the kernel inclusion condition fail. 

Therefore, the GHZ$_4$ state does not satisfy the VQMC algebraic criterion. 
\end{proof}

\subsection{W state for VQMC}

\begin{lemma}[Support--kernel duality and rank monotonicity]\label{lemma:Supp}
Let $X,Y \succeq 0$ act on the same finite-dimensional Hilbert space. Then
\begin{equation}
\mathrm{Ker}(X)\subseteq \mathrm{Ker}(Y)\quad \Longleftrightarrow \quad \mathrm{Supp}(Y)\subseteq \mathrm{Supp}(X),
\end{equation}
where $\mathrm{Supp}(Z)=\mathrm{Ker}(Z)^\perp$ denotes the support of $Z$. Consequently,
\begin{equation}\label{eq:rank-monotonicity}
\mathrm{Ker}(X)\subseteq \mathrm{Ker}(Y)\quad \Longrightarrow \quad \mathrm{rank}(X)\ \ge\ \mathrm{rank}(Y).
\end{equation}
\end{lemma}

\begin{proof}
The equivalence follows from the fact that for any subspace $S$, 
$\mathrm{Ker}(Z)=S$ if and only if $\mathrm{Supp}(Z)=S^\perp$. 
Since $\dim\mathrm{Supp}(Z)=\mathrm{rank}(Z)$ for $Z\succeq 0$, 
the inclusion of supports directly implies the rank inequality \eqref{eq:rank-monotonicity}.
\end{proof}

\begin{proposition}
Let $\rho_{ABCD} = |\mathrm{W}_4\rangle \langle \mathrm{W}_4|$ be the four-qubit W state in \eqref{eq9}, then $\rho_{ABCD}$ is a virtual quantum Markov chain (VQMC) state. That is, there exists a completely positive trace-preserving (CPTP) map $\mathcal{R}: C \rightarrow CD$ such that
\begin{equation}
\rho_{ABCD} = (\mathrm{id}_{AB} \otimes \mathcal{R})(\rho_{ABC}).
\end{equation}
\end{proposition}
\begin{proof}
Tracing out qubit $D$ from the four-qubit W state yields a mixed three-qubit state $\rho_{ABC}$ in \eqref{W:rhoABC}.
We now compute the conditional states of $AC$ upon measurement of system $C$. 
Consider first the components of $|\mathrm{W}_4\rangle$ in which the $C$ qubit is in the basis state $|0\rangle$.

\begin{equation}
|1000\rangle,\ |0100\rangle,\ |0001\rangle 
\ \Rightarrow\ \text{after tracing out $D$}:\
|100\rangle,\ |010\rangle,\ |000\rangle.
\end{equation}
Projecting onto $C$ in state $|0\rangle$ and tracing out $B$ yields
\begin{equation}
\hat{\rho}_{AC|0} = \frac{1}{4} \left( |00\rangle + |10\rangle \right)\left( \langle00| + \langle10| \right).
\end{equation}
Next, the only component of $|\mathrm{W}_4\rangle$ with the $C$ qubit in state $|1\rangle$ is $|0010\rangle$, which reduces to $|001\rangle$ in $\rho_{ABC}$. Thus,
\begin{align}
\hat{\rho}_{AC|1} &= \frac{1}{4} \, |01\rangle\langle 01| 
= \frac{1}{4} |\psi_1\rangle \langle \psi_1|, \quad \text{where } |\psi_1\rangle = |01\rangle.
\end{align}
By direct computation, the conditional operators on subsystem $BC$ are
\begin{equation}\label{eq:41}
\hat{\rho}_{BC|0} = \tfrac{1}{4}\,(|00\rangle+|10\rangle)(\langle 00|+\langle 10|), 
\qquad
\hat{\rho}_{BC|1} = \tfrac{1}{4}\,|01\rangle\langle 01|.
\end{equation}
Hence, we have
\begin{align}
\operatorname{Ker}(\hat{\rho}_{AC|0})
=
\operatorname{Ker}(\hat{\rho}_{BC|0}) &= \mathrm{span}\{ |01\rangle, |11\rangle, |00\rangle - |10\rangle \}, \\
\operatorname{Ker}(\hat{\rho}_{AC|1})
=
\operatorname{Ker}(\hat{\rho}_{BC|1}) &= \mathrm{span}\{ |00\rangle, |10\rangle, |11\rangle \}.
\end{align}

For $X=\hat{\rho}_{AC|j}$ and $Y=\hat{\rho}_{BC|j}$, the kernel-inclusion condition
\begin{equation}
\mathrm{Ker}(\hat{\rho}_{AC|j}) \subseteq \mathrm{Ker}(\hat{\rho}_{BC|j})
\end{equation}
implies, by Lemma~\ref{lemma:Supp}, that
\begin{equation}
\mathrm{rank}(\hat{\rho}_{AC|j}) \ \ge \ \mathrm{rank}(\hat{\rho}_{BC|j}).
\end{equation}
Since $\hat{\rho}_{BC|0}$ and $\hat{\rho}_{BC|1}$ are both rank-one, it follows that $\mathrm{rank}(\hat{\rho}_{AC|j}) \ge 1$. 
From the explicit forms in Eq.\eqref{eq:41}, we see that $\mathrm{rank}(\hat{\rho}_{AC|0})=\mathrm{rank}(\hat{\rho}_{AC|1})=1$. 
Therefore, for both $j=0,1$ we have
\begin{equation}
\mathrm{Ker}(\hat{\rho}_{AC|j})=\mathrm{Ker}(\hat{\rho}_{BC|j}),
\end{equation}
and in particular the inclusion $\mathrm{Ker}(\hat{\rho}_{AC|j}) \subseteq \mathrm{Ker}(\hat{\rho}_{BC|j})$ holds.

To further confirm that the four-qubit W state satisfies the VQMC condition, we explicitly construct a CPTP map $\mathcal{R}: C \to C \otimes D$ with Choi matrix
\begin{equation}
J_{CC'D} = |0\rangle\langle0|_C \otimes |\psi_0\rangle\langle \psi_0|_{C'D} 
+ |1\rangle\langle1|_C \otimes |\psi_1\rangle\langle \psi_1|_{C'D},
\end{equation}
where
\begin{equation}
|\psi_0\rangle = \tfrac{1}{\sqrt{2}}(|00\rangle + |01\rangle), 
\qquad |\psi_1\rangle = |10\rangle.
\end{equation}
This operator is positive semidefinite and satisfies
\begin{equation}
\operatorname{Tr}_{C'D}(J_{CC'D}) = I_C,
\end{equation}
thus defining a valid CPTP map $\mathcal{R}$ via the Choi–Jamiolkowski isomorphism:
\begin{equation}
\mathcal{R}(\rho) = \operatorname{Tr}_{C'}\!\left[ (\rho^T \otimes I_{C'D}) J_{CC'D} \right].
\end{equation}
Applying $\mathrm{id}_{AB} \otimes \mathcal{R}$ to the marginal state $\rho_{ABC}$ in \eqref{W:rhoABC} gives
\begin{equation}
(\mathrm{id}_{AB} \otimes \mathcal{R})(\rho_{ABC}) = \rho_{ABCD}.
\end{equation}
Therefore, the four-qubit W state indeed satisfies the VQMC condition. 
\qed
\end{proof}

\subsection{Recover four-qubit from two-qubit}

We now investigate whether a four-qubit quantum state $\rho_{ABCD}$ can be recovered directly from its two-qubit state $\rho_{AB}$ via a completely positive trace-preserving (CPTP) channel $\mathcal{R}: B \to BCD$. In other words, we ask whether
\begin{equation}
\rho_{ABCD} = (I_A \otimes \mathcal{R})(\rho_{AB})
\end{equation}
can hold for some fixed $\mathcal{R}$, independently of the particular structure of $\rho_{ABCD}$.

We expand the global state in an operator basis on qubit $A$:
\begin{equation}
\rho_{ABCD} = \sum_{i,j \in \{0,1\}} |i\rangle\langle j|_A \otimes M_{ij}^{BCD},
\end{equation}
where
\begin{equation}
M_{ij}^{BCD} \coloneqq {}_A\!\langle i| \rho_{ABCD} |j\rangle_A
\end{equation}
is an operator acting on $BCD$. Tracing out $CD$ gives
\begin{equation}
\rho_{AB} = \operatorname{Tr}_{CD} \rho_{ABCD} = \sum_{i,j} |i\rangle\langle j|_A \otimes \operatorname{Tr}_{CD}\big( M_{ij}^{BCD} \big).
\end{equation}

If such a channel $\mathcal{R}$ exists, linearity of quantum operations implies that for all $i,j$,
\begin{equation} \label{eq:necessary}
M_{ij}^{BCD} = \mathcal{R}\!\left( \operatorname{Tr}_{CD} M_{ij}^{BCD} \right).
\end{equation}
Thus, the action of $\mathcal{R}$ on the $B$-operators $\operatorname{Tr}_{CD} M_{ij}^{BCD}$ must reproduce the full $BCD$-operators $M_{ij}^{BCD}$.

We consider the W state defined as \eqref{eq9}
and let $\rho_{ABCD} = |W_4\rangle\langle W_4|$.

We group the terms by the value of qubit $A$:
\begin{align}
|W_4\rangle &= |0\rangle_A \otimes \frac{1}{2} \left( |001\rangle_{BCD} + |010\rangle_{BCD} + |100\rangle_{BCD} \right) \nonumber \\
&\quad + |1\rangle_A \otimes \frac{1}{2} |000\rangle_{BCD}.
\end{align}

From this decomposition we can write down the block operators $M_{ij}^{BCD}$ explicitly:
\begin{align}
M_{00}^{BCD} &= \frac{1}{4} \left( |001\rangle + |010\rangle + |100\rangle \right) \left( \langle 001| + \langle 010| + \langle 100| \right), \\
M_{01}^{BCD} &= \frac{1}{4} \left( |001\rangle + |010\rangle + |100\rangle \right) \langle 000|, \\
M_{10}^{BCD} &= \frac{1}{4} |000\rangle \left( \langle 001| + \langle 010| + \langle 100| \right), \\
M_{11}^{BCD} &= \frac{1}{4} |000\rangle\langle 000|.
\end{align}

We now compute their partial traces over $CD$:
\begin{align}
\operatorname{Tr}_{CD} M_{00}^{BCD} &= \frac{1}{4} \left( |0\rangle\langle 0|_B + |1\rangle\langle 1|_B + |0\rangle\langle 0|_B \right) \nonumber \\
&= \frac{1}{2} |0\rangle\langle 0|_B + \frac{1}{4} |1\rangle\langle 1|_B, \\
\operatorname{Tr}_{CD} M_{01}^{BCD} &= \frac{1}{4} |0\rangle\langle 0|_B, \\
\operatorname{Tr}_{CD} M_{10}^{BCD} &= \frac{1}{4} |0\rangle\langle 0|_B, \\
\operatorname{Tr}_{CD} M_{11}^{BCD} &= \frac{1}{4} |0\rangle\langle 0|_B.
\end{align}

Suppose a channel $\mathcal{R}$ satisfying \eqref{eq:necessary} exists. Then, in particular, we must have
\begin{align}
\mathcal{R}\!\left( \tfrac{1}{4} |0\rangle\langle 0|_B \right) &= M_{01}^{BCD} = \frac{1}{4} \left( |001\rangle + |010\rangle + |100\rangle \right) \langle 000|, \label{eq:cond01} \\
\mathcal{R}\!\left( \tfrac{1}{4} |0\rangle\langle 0|_B \right) &= M_{11}^{BCD} = \frac{1}{4} |000\rangle\langle 000|. \label{eq:cond11}
\end{align}
The left-hand sides of \eqref{eq:cond01} and \eqref{eq:cond11} are identical because the inputs are the same operator on $B$, yet the right-hand sides are \emph{different} operators on $BCD$ (one contains off-diagonal terms, the other is purely diagonal). This is impossible for a linear map $\mathcal{R}$.

Therefore, no CPTP channel $\mathcal{R}:B\to BCD$ can satisfy \eqref{eq:necessary} for all blocks $M_{ij}^{BCD}$ in the case of the $W_4$ state. This shows that, in general, it is not possible to recover $\rho_{ABCD}$ from $\rho_{AB}$ alone by acting only on subsystem $B$. Larger marginals, such as $\rho_{ABC}$, are required for recovery in the context of virtual quantum Markov chains.

\subsection{ Mixture of the 4-qubit GHZ and W states for VQMC}

We consider the convex mixture of the four-qubit GHZ and W states \eqref{eq:rhop}
where GHZ and W states are defined as \eqref{eq9}\eqref{eq:GHZ4},
and $\rho_{\mathrm{GHZ}_4} := |\mathrm{GHZ}_4\rangle\langle \mathrm{GHZ}_4|$, 
$\rho_{\mathrm{W}_4} := |\mathrm{W}_4\rangle\langle \mathrm{W}_4|$.
We examine whether the resulting state $\rho_p$ satisfies the virtual quantum Markov chain (VQMC) condition \eqref{eq16}. We now compute the $A$-side block matrices $\Theta_{B|A}$ and $\Theta_{BC|A}$ of the marginal states $\rho^{(p)}_{AB} = \operatorname{Tr}_C \rho_{ABC}^{(p)}$ and $\rho_{ABC}^{(p)}$, respectively. We consider \eqref{W:rhoABC}. For the GHZ component, tracing out subsystem $C$ gives 
\begin{equation}  
\rho_{AB}^{(GHZ)} = \operatorname{Tr}_C(\rho_{ABC}^{(GHZ)}) = \frac{1}{2}(|00\rangle\langle 00| + |11\rangle\langle 11|).
\end{equation}
For the W component we obtain
\begin{equation}
\rho_{AB}^{(W)} = \operatorname{Tr}_C(\rho_{ABC}^{(W)}) = \frac{1}{4} \left( |00\rangle\langle 00| + \sigma_{AB} \right),
\end{equation}
where $\sigma_{AB}$ denotes the partial trace of $|\phi\rangle\langle \phi|$ over $C$. Hence the convex mixture satisfies
\begin{equation}
\rho_{AB}^{(p)} = p\, \rho_{AB}^{(W)} + (1 - p)\, \rho_{AB}^{(GHZ)}.
\end{equation}
From $\rho_{AB}^{(p)}$ we extract the $A$-side block matrix  $\Theta_{B|A}^{(p)} = [ \langle i|_A \rho_{AB}^{(p)} |j\rangle_A ]_{i,j=0,1}$.
Similarly we have
\begin{equation}
\Theta_{BC|A}^{(p)} = p\, \Theta_{BC|A}^{(W)} + (1 - p)\, \Theta_{BC|A}^{(GHZ)},
\end{equation}
where each block is given by
\begin{equation}
[\Theta_{BC|A}]_{ij} := \langle i|_A \rho_{ABC}^{(p)} |j\rangle_A.
\end{equation}
These blocks are $4 \times 4$ matrices on $BC$.

For instance:
\begin{equation}
\Theta_{BC|A}^{(GHZ)} = 
\begin{bmatrix}
\frac{1}{2} |00\rangle\langle 00| & 0 \\
0 & \frac{1}{2} |11\rangle\langle 11|
\end{bmatrix},
\quad
\Theta_{BC|A}^{(W)} = 
\begin{bmatrix}
\frac{1}{4} |00\rangle\langle 00| + X & Y \\
Y^\dagger & Z
\end{bmatrix},
\end{equation}
with $X, Y, Z$ derived from $|\phi\rangle\langle\phi|$.
The VQMC criterion requires:
\begin{equation}
\mathrm{Ker}(\Theta_{B|A}^{(p)}) \subseteq \mathrm{Ker}(\Theta_{BC|A}^{(p)}).
\end{equation}
Due to the rank structure of the GHZ blocks, which are diagonal and rank-1 in each block, their kernels are large. The W blocks introduce off-diagonal components which shrink the kernels.
As $p \to 1$, $\Theta_{BC|A}^{(p)}$ approaches the fully supported W-type block matrix and the inclusion holds. However, for $p$ small (GHZ-dominated), the off-diagonal support is too small to satisfy the inclusion, and the condition fails.
Therefore, there exists a threshold value $p_c$ such that $\rho_p$ satisfies the VQMC condition if and only if $p \ge p_c$.

\subsubsection{Analytic estimate of the critical threshold}

To estimate the minimum weight $p_c$ required for the mixed state 
$\rho_p := p \rho_{\mathrm{W}_4} + (1-p) \rho_{\mathrm{GHZ}_4}$
to satisfy the virtual quantum Markov chain condition, we consider the kernel inclusion criterion:
\begin{equation}
\mathrm{Ker}(\Theta_{B|A}^{(p)}) \subseteq \mathrm{Ker}(\Theta_{BC|A}^{(p)}).
\end{equation}
This condition fails when vectors in the kernel of $\Theta_{B|A}^{(p)}$ are mapped non-trivially by $\Theta_{BC|A}^{(p)}$. Using representative vectors and tracking their images under the W and GHZ components, we find that for small $p$, the GHZ support dominates and the inclusion fails. As $p$ increases, the off-diagonal components contributed by $|\phi\rangle$ in $\rho^{(W)}_{ABC}$ begin to shrink the kernel of $\Theta_{BC|A}^{(p)}$. 

\begin{theorem}
   There exists a sharp threshold $p_c$ such that $\rho_p$ is a VQMC if and only if $p \ge p_c$. The lower bound $p_c$ satisfies
\begin{equation}
p_c \ge \frac{1}{4}.
\end{equation}
\end{theorem}
The exact value of $p_c$ can be further determined numerically by SDP feasibility checks.

A natural question is whether the set of four-qubit states satisfying the virtual quantum Markov chain (VQMC) condition is convex. That is, if two states $\rho_1, \rho_2 \in \mathcal{V}_4$, does it follow that any convex combination $\rho_\lambda := \lambda \rho_1 + (1 - \lambda) \rho_2$ also belongs to $\mathcal{V}_4$ for all $\lambda \in [0,1]$? This question arises naturally because many important sets in quantum information theory---such as separable states or PPT states---are convex. However, we will show that the VQMC condition is not preserved under convex mixing, thereby disproving this convexity conjecture.

To that end, we construct a concrete counterexample. Let $\rho_1$ be the pure four-qubit W state \eqref{eq9} and let $\rho_2$ be the classical-conditional VQMC state
\begin{equation}
\rho_2 := \frac{1}{2} \left( |00\rangle_{AB}\langle 00| \otimes |0\rangle_C\langle 0| \otimes |0\rangle_D\langle 0| + |11\rangle_{AB}\langle 11| \otimes |1\rangle_C\langle 1| \otimes |1\rangle_D\langle 1| \right).
\end{equation}
Both $\rho_1$ and $\rho_2$ satisfy the VQMC condition. For $\rho_1$, this has been verified via direct evaluation of the inclusion criterion. For $\rho_2$, the state is conditionally separable given system $C$, and can be trivially reconstructed via a recovery map that measures $C$ and appends $\rho_D^{(j)}$ accordingly.

Now consider the convex combination $\rho_\lambda = \lambda \rho_1 + (1 - \lambda) \rho_2$ with $\lambda \in (0,1)$, and examine the conditional post-measurement states given outcome $C=1$. From $\rho_1$, we have
\begin{equation}
\hat{\rho}_{AC|1}^{(1)} = |01\rangle\langle 01|, \quad \hat{\rho}_{BC|1}^{(1)} = |00\rangle\langle 00|.
\end{equation}
From $\rho_2$, we have
\begin{equation}
\hat{\rho}_{AC|1}^{(2)} = |11\rangle\langle 11|, \quad \hat{\rho}_{BC|1}^{(2)} = |11\rangle\langle 11|.
\end{equation}
The convex mixture gives
\begin{equation}
\hat{\rho}_{AC|1} = \lambda |01\rangle\langle 01| + (1 - \lambda) |11\rangle\langle 11|,\quad
\hat{\rho}_{BC|1} = \lambda |00\rangle\langle 00| + (1 - \lambda) |11\rangle\langle 11|.
\end{equation}
These conditional states have supports on orthogonal pairs, leading to the kernel spaces
\begin{equation}
\operatorname{Ker}(\hat{\rho}_{AC|1}) = \mathrm{span} \{ |00\rangle, |10\rangle \},\quad
\operatorname{Ker}(\hat{\rho}_{BC|1}) = \mathrm{span} \{ |01\rangle, |10\rangle \}.
\end{equation}
Since $|00\rangle \in \operatorname{Ker}(\hat{\rho}_{AC|1})$ but $|00\rangle \notin \operatorname{Ker}(\hat{\rho}_{BC|1})$, the inclusion condition fails:
\begin{equation}
\operatorname{Ker}(\hat{\rho}_{AC|1}) \not\subseteq \operatorname{Ker}(\hat{\rho}_{BC|1}),
\end{equation}
hence $\rho_\lambda \notin \mathcal{V}_4$ despite both endpoints belonging to $\mathcal{V}_4$. This constitutes a strict counterexample to convexity of the VQMC state space.

This result illustrates the geometric richness and structural complexity of $\mathcal{V}_4$. The failure of the VQMC condition under convex mixing arises from a delicate mismatch between support spaces of conditional marginal states, which under interpolation can cause the kernel inclusion to break down. This suggests that $\mathcal{V}_4$ is non-convex and may possess stratified or curved boundaries, unlike more familiar convex sets such as separable or PPT states.

\subsection{Choi matrix structure and interpretability} 

In cases where partial information about the recovery map is available (e.g., symmetry constraints or measurement structure), one may directly analyze the structure of the Choi matrix. This can reveal whether the state admits a classical conditioning structure or admits a measurement-based recovery mechanism. Such cases help bridge between operational interpretations and abstract recovery frameworks.

\section{Conclusion}

In this work, we have extended the framework of virtual quantum Markov chains (VQMCs) from tripartite to four-qubit systems, establishing both structural characterizations and computational tools for their verification. We demonstrated that the kernel inclusion criterion provides a necessary condition, analyzed explicit examples such as the W and GHZ states, and developed an SDP-based approach for quantifying the sampling overhead required for virtual recovery. These results reveal the non-convexity of the VQMC set and highlight open directions, including extensions to higher-dimensional systems and more efficient SDP algorithms. Beyond their theoretical value, our findings have potential applications in quantum simulation, error mitigation, and quantum communication protocols.

\section*{Acknowledgements}
Authors were supported by the NNSF of China (Grant No. 12471427), and the Fundamental Research Funds for the Central Universities (Grant Nos. ZG216S2110).

\end{document}